\documentclass[%
 reprint,
superscriptaddress,
 amsmath,amssymb,
 aps,
pra,
onecolumn, nofootinbib]
{revtex4-2}

\usepackage{graphicx} % Required for inserting images
\usepackage[a4paper, total={7in, 10in}]{geometry}
\usepackage{amsmath}
\usepackage{amsthm}
\usepackage{amssymb}
\usepackage{mathtools}
\usepackage{bm}
\usepackage{comment}
\usepackage{mathalpha}
\usepackage{xcolor}
\usepackage{amsfonts}
\usepackage[braket]{qcircuit}
\usepackage{nicematrix}

\usepackage{subcaption}
\usepackage[labelfont=bf,
   justification=raggedright,
   format=plain]{caption}

\usepackage[T1]{fontenc}
\usepackage{hyperref}
\usepackage{complexity}
\usepackage[capitalize,nameinlink]{cleveref}
\usepackage[linesnumbered,ruled,vlined]{algorithm2e}
\usepackage{algpseudocode}
\usepackage{tikz}
\usepackage{nccmath}
\usepackage{thm-restate}

\hypersetup{
    colorlinks=true, % Enable colors instead of boxes
    linkcolor=black,  % Color for internal links
    citecolor=blue, % Color for citations
}
\usetikzlibrary{automata, positioning, arrows}

\newtheorem{theorem}{Theorem}

\newtheorem{lemma}[theorem]{Lemma}

\begin{document}
\title{Ground State Preparation via Dynamical Cooling}

\author{Danial Motlagh}
\email{danial.motlagh@xanadu.ai}
\affiliation{Xanadu, Toronto, ON, M5G2C8, Canada}%Lines break automatically or can be forced with \\

\author{Modjtaba Shokrian Zini}
\affiliation{Xanadu, Toronto, ON, M5G2C8, Canada}

\author{Juan Miguel Arrazola}
\affiliation{Xanadu, Toronto, ON, M5G2C8, Canada}

\author{Nathan Wiebe}
% \author[1]{Danial Motlagh}
% \author[1]{Modjtaba Shokrian Zini}
% \author[1]{Juan Miguel Arrazola}
% \author[2, 3, 4]{Nathan Wiebe}
% \affil[1]{Xanadu, Toronto, ON, M5G 2C8, Canada}
\affiliation{
University of Toronto, Department of Computer Science, Toronto ON, Canada
}%
\affiliation{
Pacific Northwest National Laboratory, Richland WA, USA
}%
\affiliation{
Canadian Institute for Advanced Research, Toronto ON, Canada
}%
% \date{}
% \begin{document}

%%%% CHECK OUT THIS FOR THE ABSTRACT: https://www.nature.com/documents/nature-summary-paragraph.pdf or alternatively https://chemistrycommunity.nature.com/posts/43071-how-to-write-an-abstract

\begin{abstract}
Quantum algorithms for probing ground-state properties of quantum systems require good initial states. Projection-based methods such as eigenvalue filtering rely on inputs that have a significant overlap with the low-energy subspace, which can be challenging for large, strongly-correlated systems. This issue has motivated the study of physically-inspired dynamical approaches such as thermodynamic cooling. In this work, we introduce a ground-state preparation algorithm based on the simulation of quantum dynamics. Our main insight is to transform the Hamiltonian by a shifted sign function via quantum signal processing, effectively mapping eigenvalues into positive and negative subspaces separated by a large gap. This automatically ensures that all states within each subspace conserve energy with respect to the transformed Hamiltonian. Subsequent time-evolution with a perturbed Hamiltonian induces transitions to lower-energy states while preventing unwanted jumps to higher energy states. The approach does not rely on a priori knowledge of energy gaps and requires no additional qubits to model a bath. Furthermore, it makes $\Tilde{\mathcal{O}}(d^{\,3/2}/\epsilon)$ queries to the time-evolution operator of the system and $\Tilde{\mathcal{O}}(d^{\,3/2})$ queries to a block-encoding of the perturbation, for $d$ cooling steps and an $\epsilon$-accurate energy resolution. Our results provide a framework for combining quantum signal processing and Hamiltonian simulation to design heuristic quantum algorithms for ground-state preparation.   
\end{abstract}
\maketitle

\section{Introduction}

Computing ground-state properties of quantum systems is an important challenge in quantum chemistry and condensed-matter physics. This has motivated the use of quantum computers as a potential avenue to obtain improved solutions~\cite{aspuru2005simulated, ge2019faster,lanyon2010towards, lin2020near, o2019quantum, veis2010quantum, lin2022heisenberg}. However, preparing ground states and estimating properties such as the ground-state energy presents difficulties even for quantum computing~\cite{lee2023evaluating}. Quantum algorithms like quantum phase estimation (QPE) \cite{kitaev1995quantum, aspuru2005simulated} can compute the energy of an input state in polynomial time, but it remains challenging to prepare approximate ground states of many-body Hamiltonians to be used as an initial state for QPE~\cite{fomichev2023initial}. This difficulty can be traced back to the fact that in its full generality, ground-state energy estimation is \QMA-complete \cite{aharonov2009power, kempe2006complexity, kitaev2002classical, cubitt2016complexity, gottesman2009quantum}.\\

%Despite the theoretical obstacles in proving algorithmic guarantees for efficient preparation of ground states, there are still arguments to be made as to why this procedure does not have to take exponential time in most cases that we care about, namely, when the Hamiltonian of interest is related to the dynamics of a physical system. A notable piece of empirical evidence supporting this argument is the fact that physical systems in the real world naturally settle into their ground state in finite time by evolving according to their relevant physical dynamics. 
Unlike the ground-state problem, simulating time evolution under a local Hamiltonian is tractable on quantum computers~\cite{feynman2018simulating, lloyd1996universal, nielsen2002universal, childs2012hamiltonian, wiebe2011simulating, haah2021quantum, schleich2024chemically}, and computationally is in the class \BQP ~of problems that can be solved in polynomial time on a quantum computer with high probability. Furthermore, no general sub-exponential-time classical algorithm is known for simulating quantum dynamics. 
This raises the question: is it possible to simulate the dynamical process by which physical systems reach their ground state to perform more efficient state preparation? While QMA-hardness indicates that not every ground state can be prepared this way, many examples of practical importance, such as real molecules and materials, are known to thermalize efficiently. By mimicking the natural processes that bring physical systems to their ground state, we may be able to prepare low-energy states more efficiently than traditional approaches to quantum ground state preparation. \\

Recent work~\cite{fomichev2023initial} has argued that the leading methods for initial state preparation in quantum chemistry are those based on encoding classical wavefunctions on a quantum computer. In practice, when dealing with large and highly-correlated systems, classical methods can struggle to prepare initial states with significant overlap with the ground state, or even with a low-energy subspace. This presents a challenge as many of the traditional state preparation methods in quantum computing have hitherto been ``static projection'' methods, whose success probability is determined by the overlap between the input state and the desired eigenstate~\cite{lin2020near, dong2022ground}. Therefore, there is motivation to develop fully quantum algorithms for state preparation that employ dynamic and adaptive approaches that go beyond the limitations of classical initial states.\\

% Many of the traditional state preparation methods in quantum computing have hitherto been ``static projection'' methods, whose success probability is determined by the overlap between an initial state and the desired eigenstate~\cite{lin2020near, dong2022ground}. In practice, these methods are limited when the initial state has a small overlap with the ground state, or even with a low-energy subspace, a scenario that often arises with highly-correlated systems. Recent work~\cite{fomichev2023initial} has argued that the leading methods for initial state preparation in quantum chemistry are actually those based on encoding classical wavefunctions on a quantum computer. There is a motivation to develop fully quantum algorithms employing dynamic and adaptive approaches to go beyond the limitations of classical initial states. 

Among these, dissipative state preparation methods have gained traction~\cite{verstraete2009quantum, kraus2008preparation, su2020quantum, polla2021quantum, shtanko2021algorithms, cubitt2023dissipative, ding2023single, pocrnic2023quantum, kastoryano2023quantum}, taking inspiration from thermodynamic processes such as cooling via a bath. Many of these techniques use ancillary qubits to act as external heat baths that interact with the system to siphon out energy and guide it toward its ground state. By initially placing the bath in its ground state and using conservation of energy in the weak interaction limit, it is possible to ensure that the only allowed transitions are those that decrease or maintain the energy of the system. However, this also means that the only possible transitions are those that decrease the system's energy by an available energy level in the bath. This requires previous knowledge of the gaps between the eigenstates of the Hamiltonian.\\

Our work contributes to this emerging paradigm with a quantum algorithm that offers several improvements compared to previous work. First, it does not require a priori knowledge of energy differences between eigenstates of the Hamiltonian. Second, no additional qubits are needed to represent a bath -- all transformations occur exclusively in the Hilbert space of the system. Finally, our algorithm can succeed in preparing approximations of the ground state even if the overlap of the initial state is zero (or close to zero), making our approach suitable to instances where classical methods have difficulty in ensuring large overlaps. \\

The quantum algorithm uses a combination of quantum signal processing (QSP) and dynamical simulation under a perturbed Hamiltonian to engineer transformations that systematically and iteratively reduce the energy of the system (with high probability). After coarsely estimating the energy $\mathcal{E}$ of the input state, the Hamiltonian $H$ is transformed by a shifted sign function using QSP, with the shift depending on the energy estimate. The result is a new operator $H_{\text{sign}}=\text{sign}(H-\mathcal{E})$, whose eigenvalues are split into positive and negative subspaces separated by a large gap. The system is then time-evolved under a new perturbed Hamiltonian $\tilde{H}=H_{\text{sign}} + \sqrt{\delta}A/2$, where $A$ is a Hermitian perturbation operator, and $\delta>0$ is a parameter that sets an upper bound on the probability of transitioning to a state of higher energy. This dynamic process allows transitions to favourable lower-energy states without knowledge of their energy gaps, and forbids transitions to unwanted higher-energy ones, leading to an average reduction of energy. Repeating these steps multiple times leads to increasingly better approximations of the ground state.\\

The paper is structured as follows. In~\cref{sec:Prelim} we summarize several preliminary results and definitions that will be used throughout~\cref{sec:Algo}. After defining the access model in~\cref{subsec:access_model}, we introduce the quantum algorithm in~\cref{subsec:incoherent} which uses only a single ancilla qubit for quantum phase estimation (QPE) purposes. To also make our framework compatible as a subroutine within larger workflows that may include amplitude amplification steps, in \cref{subsec:coherent} we provide a coherent (unitary) version with an added logarithmic space overhead per iteration. Lastly, we perform a complexity analysis of the algorithms in~\cref{subsec:Analysis} and discuss a few points on the choice of interaction in~\cref{subsec:choiceofA}.\\

\section{Preliminaries and Definitions}\label{sec:Prelim}
This section provides the preliminary results and definitions needed to understand the algorithm presented in the following section. The main techniques that we employ are generalized quantum signal processing (GQSP)~\cite{motlagh2023generalized} and optimal polynomial approximations to the sign function. We then introduce a few definitions and notations that will be used in the rest of the paper. \\

The results of GQSP state that any polynomial $P(x) = \sum_{n=-k}^m a_n x^n$ of a unitary operator $U$ can be implemented using only $m$ applications of controlled-$U$, $k$ applications of controlled-$U^\dagger$, and $m+k+1$ single-qubit rotations, as long as $|P|\leq1$ on the complex unit circle. We use 
\begin{equation}
R(\theta, \phi, \lambda) = \begin{bmatrix} e^{i(\lambda+\phi)}\cos(\theta) & e^{i\phi}\sin(\theta) \\
e^{i\lambda}\sin(\theta) & -\cos(\theta)  \\ \end{bmatrix} \otimes I\label{eq:R},
\end{equation}
to represent a fully parametrized SU(2) rotation on the ancillary qubit used in QSP. In the following, we write $R(\theta, \phi)$ in place of $R(\theta, \phi, 0)$ for succinctness. 
\begin{theorem}\label{thm:GQSP}
     For any polynomial $P(x)\in \left\{ \sum_{n=-k}^{m} a_n x^n \mid a_n \in \mathbb{C}\right\}$, if for all $\, x\in\mathbb{R}$ it holds that $\, |P(e^{ix})|^2 \leq 1$, then there exist angles $ \vec{\theta}, \vec{\phi} \in \mathbb{R}^{k+m+1}$ and a parameter $ \lambda \in \mathbb{R}$ such that:\\
\[
\begin{bmatrix}
P(U) & .\>\>\>  \\
.\>\>\> & .\>\>\> \\
\end{bmatrix} = \left( \prod_{j=1}^{k} R(\theta_{d+j}, \phi_{d+j}) \begin{bmatrix}
\mathbb{I} & 0  \\
0 & U^\dagger \\
\end{bmatrix} \right)\left( \prod_{j=1}^{m} R(\theta_{j}, \phi_{j}) \begin{bmatrix}
U & 0  \\
0 & \mathbb{I} \\
\end{bmatrix} \right) R(\theta_0, \phi_0, \lambda).
\]
\end{theorem}
\begin{proof}
    The proof follows by combining Theorem 6 and Corollary 5 of Ref.~\cite{motlagh2023generalized}.
\end{proof}
Next, we quote the following regarding the optimal polynomial approximation to the sign function.
\begin{lemma}[\text{\cite[Lemma 25]{QSVT}}]\label{lem:Poly_Approximation_of_sign} $\forall\delta, \epsilon \in (0,1)$, there exists an efficiently computable odd polynomial \(P(x, \epsilon, \delta) \in \mathbb{R}[x]\) of degree \(\mathcal{O}\left(\frac{1}{\epsilon} \log\left(\frac{1}{\delta}\right)\right)\) such that
\begin{enumerate}
\item \(\forall x \in [-1, 1]\), \(|P(x)| \leq 1\).
\item \(\forall x \in [-1, \frac{-\epsilon}{2}] \cup [\frac{\epsilon}{2}, 1]\), \(|P(x) - {\rm sign}(x)| \leq \delta\).
\end{enumerate}
\end{lemma}
Using the polynomial approximation results above, we derive a Fourier approximation version with the same asymptotic complexity to use with the results of GQSP.

\begin{lemma}\label{lem:Fourier_Approximation_of_sign}
(Fourier approximation of the sign function). For all \(\delta, \epsilon \in (0,1)\), there exists an efficiently computable odd polynomial \(S(x, \epsilon, \delta) \in \mathbb{C}[x, x^{-1}]\) of degree \(\mathcal{O}\left(\frac{1}{\epsilon} \log\left(\frac{1}{\delta}\right)\right)\), such that
\begin{enumerate}
\item \(\forall x \in [-\pi, \pi]\), \(|S(e^{ix})| \leq 1\).
\item \(\forall x \in [-\pi+\frac{\epsilon}{2}, \frac{-\epsilon}{2}] \cup [\frac{\epsilon}{2}, \pi-\frac{\epsilon}{2}]\), \(|S(e^{ix}) - {\rm sign}(x)| \leq \delta\).
\end{enumerate}
\end{lemma}
\begin{proof}
    Take $P(x, \frac{\epsilon}{2}, \delta)$ from \cref{lem:Poly_Approximation_of_sign} and define the polynomial $S$ that can efficiently be computed from $P$ with deg($S$) = deg($P$) as $S(e^{ix}) =  P(\frac{e^{ix}-e^{-ix}}{2i}) = P(\sin x)$. Then notice that if $|x|> \frac{\epsilon}{4}$, we have $|\sin x| > \frac{\epsilon}{2}$ and ${\rm sign}(\sin x) = {\rm sign}(x)$ on the interval $(-\pi, \pi)$. It follows that $S(x, \epsilon, \delta)$ satisfies all the required conditions, namely $\forall x \in [-\pi, \pi]$,
    \begin{equation}
        |S(e^{ix})| \leq 1,
    \end{equation}
    and $\forall x \in [-\pi+\frac{\epsilon}{2}, \frac{-\epsilon}{2}] \cup [\frac{\epsilon}{2}, \pi-\frac{\epsilon}{2}]$
    \begin{equation}
        |S(e^{ix}) - {\rm sign}(x)| \leq \delta.
    \end{equation}
\end{proof}
We now define notation that will be used throughout this work. For a projector $\Pi$, we denote the reflection about the subspace defined by $\Pi$ as
\begin{align}\label{def:reflection}
        {\rm REF}(\Pi) := \mathbb{I}- 2\,\Pi.
\end{align}
For a Hamiltonian $H = \sum_j \lambda_j \ket{\lambda_j}\bra{\lambda_j}$, and a threshold $\mathcal{E}\in \mathbb{R}$, we denote the projector onto the subspace with energies below $\mathcal{E}$ by
\begin{align}\label{def:proj_threshold}
    \Pi_{<\mathcal{E}}(H) := \sum_{\lambda_j <\mathcal{E}} \ket{\lambda_j}\bra{\lambda_j} .
\end{align}

At this point, we have laid out all the preliminary results and definitions needed for the algorithm. The next definition is needed for the coherent version. For an $m$-qubit Hamiltonian $H$, we define $\textnormal{SHIFT}_n(H)$ as an $(n+m)$-qubit operator that shifts the spectrum of $H$ by the value of the bitstring in the auxiliary $n$-qubit register. Formally,
\begin{align}\label{def:Shift_n}
    \textnormal{SHIFT}_n(H) := \sum_{j=0}^{2^n - 1} \ket{j}\bra{j}\otimes \left(H - \frac{j\cdot2\pi}{2^n}\right),
\end{align}
where $\ket{j}$ stores the bitstring representation of integer $j$. When clear from context, we write $\textnormal{SHIFT}(H)$ instead of $\textnormal{SHIFT}_n(H)$. We show an efficient implementation of the time evolution by $\textnormal{SHIFT}_n(H)$ assuming efficient time evolution exists for $H$.
\begin{lemma}\label{lem:Shift_exponential}
Given a Hamiltonian $H = \sum_j \lambda_j \ket{\lambda_j}\bra{\lambda_j}$, and $n\in \mathbb{N^+}$, we can implement $e^{i\textnormal{SHIFT}_n(H)}$ using one application of $e^{iH}$ and $n$ single qubit rotations. More specifically:
\[
e^{i\textnormal{SHIFT}_n(H)} = \left[\bigotimes_{m=0}^{n-1} \textnormal{Phase}\left(\frac{2^m\cdot2\pi}{2^n}\right)\right] \otimes e^{iH},
\]
where $\textnormal{Phase}(\theta) = \begin{bmatrix} 1 & 0 \\ 0 & e^{-i\theta} \end{bmatrix}$.
\end{lemma}
\begin{proof}
    Writing $j = \sum_{m=0}^{n-1} j_m \cdot 2^{m}$, where $j_m$ represents the $m$-th bit in the bitstring representation of $j$, we have
    \begin{align}
        \left[\bigotimes_{m=0}^{n-1} \textnormal{Phase}\left(\frac{2^m\cdot2\pi}{2^n}\right)\right] \otimes e^{iH} \cdot(\ket{j}\otimes \mathbb{I}) &= \left[\bigotimes_{m=0}^{n-1} \textnormal{Phase}\left(\frac{2^m\cdot2\pi}{2^n}\right) \ket{j_m}\right] \otimes e^{iH}\nonumber\\
        %&= \left[\bigotimes_{m=0}^{n-1} e^{-ij_m \cdot 2^{m}\cdot2\pi/2^n} \ket{j_m}\right] \otimes e^{iH}\nonumber\\
        &= e^{-i \sum_{m=0}^{n-1}j_m \cdot 2^{m}\cdot2\pi/2^n} \ket{j} \otimes e^{iH}\nonumber\\
       % &= e^{-i j \cdot2\pi/2^n} \ket{j} \otimes e^{iH}\nonumber\\
        &= \ket{j} \otimes e^{i\left(H - (j\cdot2\pi)/2^n\right)}.
    \end{align}
\end{proof}

% DISSIPATION FRAMEWORK

\section{Dynamical Cooling for Ground State Preparation}\label{sec:Algo}

This section describes the state preparation quantum algorithm that constitutes the main result of this work. We begin by defining the input parameters and access model, which involves assumptions about oracle access to certain operations. We then describe the main quantum algorithm and its coherent version, ending with an analysis of error guarantees and runtime complexity.

\vspace{-1.2cm}

\subsection{Input Parameters and Access Model}\label{subsec:access_model}
We list the parameters and conditions in the framework. First, we view the system Hamiltonian $H$ as an input parameter, which is assumed to be normalized such that $\|H\|\leq 1$. We assume oracular access to $e^{iH}$, but the algorithm can easily be modified for block-encoding access to $H$. We impose no condition on the initial state $\ket{\psi_0}$, although its quality can impact the performance of the algorithm.
We also assume as input a Hermitian perturbation term $A$ satisfying $\|A\|\leq 1$, acting on the Hilbert space of the system. We assume block-encoding access $\bra{0}U_A\ket{0} = A$, but the algorithm can easily be modified for access to $e^{iA}$ instead.\\

We also require a precision parameter $\epsilon\in (0, 1)$ that specifies the resolution used to distinguish states with different energies. This is relevant when talking about states of higher energy; specifically, we consider it a failure to transition to a state with energy at least $\epsilon$ greater than the input state.
We further require another parameter $\delta\in (0, 1)$ that specifies an upper bound on the probability of failure in each step. Given the iterative nature of the algorithm, we could equivalently replace the $\delta$ parameter with another parameter $d$ denoting the maximum number of iterations, and set $\delta = 1/d$. This ensures that the algorithm has a constant probability of success after $d$ iterations (more in \cref{subsec:Analysis}). With this setup, we are now ready to introduce the quantum algorithm in \cref{subsec:incoherent}, followed by a coherent version in \cref{subsec:coherent}.

\vspace{-1.2cm}
\subsection{Algorithm}\label{subsec:incoherent}
Our algorithm starts by performing an iterative QSP-based phase estimation~\cite{dong2022ground} on the input state $\ket{\psi_0}$ with accuracy $\epsilon$ and probability of success $\delta$. This projects our state into an eigenstate $\ket{\lambda_k}$ (or a linear combination of eigenstates) with energy $\mathcal{E}_k\pm \epsilon/2$. Using iterative methods~\cite{dong2022ground}, it is possible to use just a single extra qubit for the QPE protocol. We then employ QSP to implement the unperturbed Hamiltonian
\begin{equation}
    H_{\text{sign}} = S(H-\mathcal{E}_k-\epsilon, \epsilon,\delta),
\end{equation}
where \(S(x, \epsilon, \delta)\) is the approximate sign function defined according to~\cref{lem:Fourier_Approximation_of_sign}, and the additional shift by $\epsilon$ is to account for the transition interval of the approximation at the discontinuous point of the sign function. Notice that all eigenstates of $H$ with energy at most $\mathcal{E}_k + \epsilon/2$ have energy in $[-1, -1+\delta]$ with respect to $H_{\text{sign}}$, and all states with energy at least $\mathcal{E}_k + 3\epsilon/2$ now have energy in $[1-\delta, 1]$. Hence, making use of this large gap, due to conservation of energy in the weak interaction limit, we can ``forbid'' transitions to states of higher energy. To do so, we define a perturbed Hamiltonian
\begin{equation}
    \tilde{H} = H_{\text{sign}} + \frac{\sqrt{\delta}}{2}A,
\end{equation}
and evolve under this operator for time  $t = \pi\cdot\left\lceil \frac{1}{\pi \sqrt{\delta}}\right\rceil\sim\mathcal{O}(\frac{1}{\sqrt{\delta}})$ to obtain the output state $\ket{\phi} = e^{-i\tilde{H}t}\ket{\lambda_k}$.\\

Based on \cref{thm:hamLeak,thm:how_it_be} discussed in the next section, except for a leakage probability $\delta$, the output state $\ket{\phi} \approx e^{-i\Pi_{<\mathcal{E}_k}(A/2)\Pi_{<\mathcal{E}_k}}\ket{\lambda_k}$ has a higher overlap with the lower energy subspace than the original state $\ket{\lambda_k}$. Hence, its energy is decreased on average. We then perform another round of QPE to get an updated energy estimate for the next cooling step. This process is repeated until a stopping condition is met; for example reaching a maximum number of steps, observing an energy below a target value, or many consecutive steps failing to decrease the energy. We describe the quantum algorithm in full detail below.

\pagebreak

\begin{figure}[t]
\centering
\includegraphics[width=0.85\textwidth]{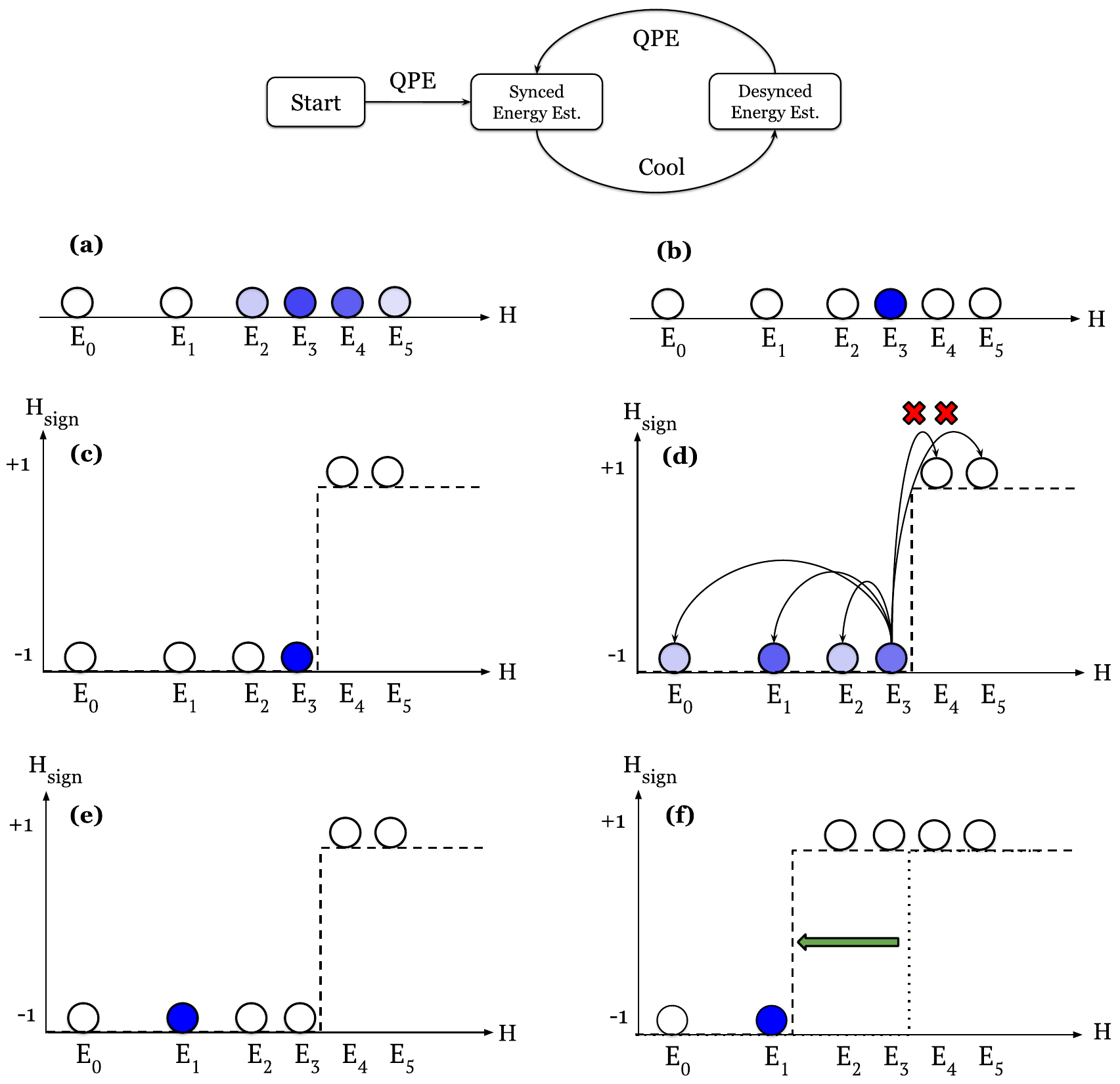}
\caption{Visual representation of how a state might evolve within the quantum algorithm. In this schematic, the horizontal axis represents energy with respect to $H$, the vertical axis denotes energy with respect to $H_{\text{sign}}$, and circles are used to denote eigenstates where the magnitude of their amplitude is represented by color intensity. \textbf{(a)} We start from an initial state which may be a superposition of eigenstates. \textbf{(b)} We perform a QPE protocol that projects onto an approximate eigenstate and gives an estimate of the state's energy, in this case $E_3$.  \textbf{(c)} We define $H_{\text{sign}}$ using the energy estimate. With respect to this operator, states are split into a negative subspace with eigenvalue $-1$, and a positive subspace with eigenvalue $+1$. \textbf{(d)} Time evolution under a perturbed Hamiltonian $\tilde{H}$ leads to transitions in the low-energy subspace, while transitions to higher energy states are forbidden due to energy conservation. \textbf{(e)} We then perform another round of QPE, in this case leading to a lower energy estimate. \textbf{(f)} A new $H_{\text{sign}}$ is defined that is further shifted toward lower energies. }
\label{fig:complete}
\end{figure}

\textbf{Quantum algorithm for ground-state preparation using dynamical cooling}

\begin{enumerate}
    \item Given an input Hamiltonian $H$, an input state $\ket{\psi_0}$, a maximum number of iterations $d$, and an energy estimation precision $\epsilon$, perform one round of QPE with accuracy $\epsilon$ and probability of success $1/d$. Denote the outcome energy estimate as $\mathcal{E}_k$ and the output state as $\ket{\lambda_k}$. 
    \item Given a failure probability parameter $\delta = 1/d$, use QSP to block-encode the transformed Hamiltonian $H_{\text{sign}}=S(H-\mathcal{E}_k-\epsilon, \epsilon,\delta)$ as in \cref{lem:Fourier_Approximation_of_sign}, leveraging oracle access to $e^{iHt}$. Note that the shift depends on the energy estimate from the previous step.
    \item Implement a block-encoding of the perturbed Hamiltonian $\tilde{H}=H_{\text{sign}}+(\sqrt{\delta}/2)A$ via a linear combination of unitaries subroutine~\cite{childs2012HamSimLCU} on the block-encodings of each component. Use QSP to evolve the state $\ket{\lambda_k}$ under $\tilde{H}$ for time $t = \pi\cdot\left\lceil \frac{1}{\pi \sqrt{\delta}}\right\rceil$, producing the output state $\ket{\psi_1}=e^{-i\tilde{H}t}\ket{\lambda_k}$. Use $\ket{\psi_1}$ as the input for the next step.
    \item Repeat steps 1 to 3 until a stopping condition is met. 
\end{enumerate}

A high-level schematic overview of the quantum algorithm is shown in \cref{fig:complete}. The key strategy is the transformation of the system Hamiltonian by the shifted sign function. This results in a splitting of the Hilbert space into two effective subspaces corresponding to negative and positive eigenvalues, after the shift by $\mathcal{E}_k$. Within each subspace, the spectrum is approximately flat --- a property that allows transitions to any eigenstate in the negative subspace during time evolution, while ensuring conservation of energy. Moreover, the large gap between negative and positive subspaces forbids undesired transitions to states of higher energy, a failure that is bounded by the parameter $\delta$ --- a lower chance of failure requires longer time evolution. As a result, the output state after time evolution has lower average energy than it started with, leading to cooling towards the ground state. We analyze error guarantees and runtime complexity in \cref{subsec:Analysis}. 

\subsection{Unitary Version}\label{subsec:coherent}
We now describe a coherent version of the algorithm that doesn't require any mid-circuit measurements, designed for use as a subroutine within larger workflows that might include an amplitude amplification procedure. The arising difficulty is that without measurements in QPE, the input state prior to transforming the Hamiltonian with QSP can generally be any linear combination of eigenstates at different energies. This means that $H_{\text{sign}}$ needs to be an approximation to the sign function shifted by different energies in superposition.\\

To address this, we must store energy estimates in superposition in a new quantum register, which can be done via a coherent QPE subroutine, i.e., a QPE without measurements. The size of each register scales logarithmically with accuracy. This QPE must have accuracy $\epsilon$ and probability of success $1-\delta$, which is achieved most efficiently using a QSP-based approach similar to the one in~\cite{martyn2021grand}. Thus, we define our unperturbed Hamiltonian over the combined space of the system and energy register as
\begin{equation}
    H_{\text{sign}} = S(\textnormal{SHIFT}(H)-\epsilon, \epsilon,\delta),\label{eq:H0def}
\end{equation}
where \textnormal{SHIFT}($H$) in~\cref{eq:H0def} is defined according to \cref{def:Shift_n}. We also extend $A$ by the identity over the energy register and define the perturbed Hamiltonian as
\begin{equation}
    \tilde{H} = H_{\text{sign}} + \frac{\sqrt{\delta}}{2}(\mathbb{I} \otimes A).
\end{equation}
The algorithm then follows as before, except that now we have to add a new energy register for each iteration; this only has a logarithmic overhead.

\subsection{Analysis}\label{subsec:Analysis}

As stated in previous sections, a key feature of the algorithm is that time evolution under the perturbed Hamiltonian $\tilde{H}$ will approximately retain the state inside the low-energy subspace. This claim is summarized in the following theorem, which we prove in \cref{appendix:first}.

\begin{restatable}[Evolution Leakage Bound]{theorem}{leaking}
\label{thm:hamLeak}%
Let $A$ be a Hermitian operator with $\|A\|\leq 1$ and let $\Pi$ be a projector. Let $\delta \in (0,1)$ and $t = \pi\cdot\left\lceil \frac{1}{\pi \sqrt{\delta}}\right\rceil$. Define $\tilde{H} = {\rm REF}(\Pi) + \frac{\sqrt{\delta}}{2}A$. Then:
\[
\left\|\left(\mathbb{I} - \Pi\right)e^{-i\tilde{H}t}\,\Pi\right\|^2 \leq \delta.
\]
\end{restatable}

Taking $\text{sign}(H-\mathcal{E}) = {\rm REF}(\Pi_{<\mathcal{E}}(H))$, the above result states that the probability of transitioning to states of higher energy is upper bounded by $\delta\sim\mathcal{O}(1/\sqrt{t})$. Furthermore, we show that a successful step evolves the state, up to an error smaller than $\delta$, according to the transformation $e^{-i\Pi_{<\mathcal{E}}(A/2)\Pi_{<\mathcal{E}}}$ determined by the choice of $A$. This is shown in the following theorem, proven in \cref{appendix:second}.

\begin{restatable}[Effective Evolution]{theorem}{dissipate}
\label{thm:how_it_be}%
Let $A$ be a Hermitian operator with $\|A\|\leq 1$ and let $\Pi$ be a projector. Let also $\delta \in (0,1)$ be a failure probability and let $t = \frac{1}{\sqrt{\delta}}$ the time of evolution. Define $\tilde{H} = {\rm REF}(\Pi) + \frac{\sqrt{\delta}}{2}A$. Then:
\[
\left\|\Pi e^{-i\tilde{H}t}\,\Pi - e^{-i\Pi(A/2)\Pi}\,\Pi\right\|^2 \leq \delta.
\]
\end{restatable}

Lastly, the following theorem gives the complexity of running the algorithm for $d$ successful iterations with an $\epsilon$-accurate energy resolution.
\begin{theorem}[Complexity Analysis]\label{claim:com_analysis}
    Running the iterative algorithm with an $\epsilon$-accurate energy resolution for $d$ iterations with a constant probability of success requires $\Tilde{\mathcal{O}}\left(d^{\,3/2}/\epsilon\right)$ queries to $e^{iH}$, and $\Tilde{\mathcal{O}}(d^{\,3/2})$ queries to $U_A$. Furthermore, the incoherent version requires $\mathcal{O}\left(1\right)$ ancilla qubits and the coherent version requires $\mathcal{O}\left(d\log(1/\epsilon)\right)$ ancilla qubits.
\end{theorem}
\begin{proof}
    First notice that the approximation $S(x, \epsilon, \delta)$ to the sign function requires a Fourier series of degree $\mathcal{O}(\log(1/\delta)/\epsilon)$ according to \cref{lem:Fourier_Approximation_of_sign}, and thus can be implemented via QSP using $\mathcal{O}(\log(1/\delta)/\epsilon)$ calls to $e^{iH}$ and its inverse according to \cref{thm:GQSP}. Furthermore, we time-evolve the perturbed Hamiltonian for $t\sim\mathcal{O}(1/\sqrt{\delta})$. 
Lastly, an $\epsilon$-accurate QPE protocol with a probability of success $1-\delta$ can be implemented with $\mathcal{O}(\log(1/\epsilon)\log(1/\delta)/\epsilon)$ queries to $e^{iH}$ using a QSP-based approach~\cite{dong2022ground, martyn2021grand}. 
Putting it all together,  the total cost of our algorithm per iteration is $\Tilde{\mathcal{O}}\left(\mfrac{1}{\sqrt{\delta}\cdot\epsilon}\right)$ queries to $e^{iH}$, plus $\Tilde{\mathcal{O}}\left(\mfrac{1}{\sqrt{\delta}}\right)$ queries to $U_A$, up to logarithmic factors. \\

By setting $\delta=1/d$, where $d$ is the maximum number of iterations, the probability of success in each iteration is lower bounded by $1-1/d$ (\cref{thm:hamLeak}), and the probability of all $d$ iterations succeeding is lower bounded by $(1-1/d)^d\geq 1/4$ for $d\geq 2$. 
Therefore, the total cost of carrying out the algorithm for $d$ iterations with a constant probability of success is $\Tilde{\mathcal{O}}\left(d^{\,3/2}/\epsilon\right)$ queries to $e^{iH}$, and $\Tilde{\mathcal{O}}(d^{\,3/2})$ queries to $U_A$.
\end{proof}

\subsection{Choosing the perturbation operator}\label{subsec:choiceofA}
In the description of the algorithm, the perturbation operator $A$ was left arbitrary, but some concrete choice of $A$ must be taken based on the problem of interest. Further, the convergence rate of our algorithm is dependent on $A$, meaning that the overall performance of preparing a low-energy state relies on adequately choosing this operator. There are many ways that we could use prior knowledge about a system to choose the perturbation term to facilitate cooling.  One approach would be to take inspiration from physical thermalization processes. In particular, one could choose $A$ to be a system-bath interaction term, but restricted to the Hilbert space of the system. For example, for light-matter interaction within the dipole approximation given by $-\Vec{\mu}\cdot \Vec{E}$, where $\Vec{\mu}$ is the electric dipole moment and $\Vec{E}$ is the electric field vector , one would set $A$ to be the dipole operator
\begin{equation}
    A = \Vec{\mu}.
\end{equation}
Then the cooling step of the algorithm would (roughly) correspond to photon emission via interaction with a light field that has available modes at any frequency, analogous to the process of spontaneous emission. One could similarly apply the same idea to vibrational relaxation or any other cooling mechanism. It is however difficult to argue how long the thermalization process would take for such physically motivated choices of the perturbation operator. Below, we consider a different strategy by choosing a completely random interaction and bound the resources needed to cool to the ground state; albeit at an expected cost that is exponential in the number of iterations.\\

For a more mathematical criteria of choosing $A$, one can make the first order approximation 
\begin{equation}
    e^{-i\Pi_{<\mathcal{E}}(A/2)\Pi_{<\mathcal{E}}} = \mathbb{I} -  i\Pi_{<\mathcal{E}}(A/2)\Pi_{<\mathcal{E}} + \mathcal{O}(\|A\|^2).
\end{equation}
From this we can write the transition matrix

\begin{align}
    T_{ij} &:=  |\bra{\lambda_i} e^{-i\Pi_{\le \mathcal{E}}(A)\Pi_{\le \mathcal{E}}} \ket{\lambda_j}|^2\nonumber\\
    &=\begin{cases}
        |\bra{\lambda_i}A\ket{\lambda_j}|^2/4 + \mathcal{O}(\|A\|^3), & \lambda_i \leq \lambda_j \wedge i\neq j\\
         1-\sum_{\lambda_k \leq \lambda_j} |\bra{\lambda_k}A\ket{\lambda_j}|^2/4 +\mathcal{O}(\|A\|^3), & i =j\\
        0, & \lambda_i > \lambda_j.\\
    \end{cases}
\end{align}
Hence within this approximation, given an initial state $\ket{\lambda_j}$, the probability that the cooling step lowers the energy is given by
\begin{align}
    \sum_{i: \lambda_i < \lambda_j} T_{ij} = \sum_{i: \lambda_i < \lambda_j} \frac{|\bra{\lambda_i}A\ket{\lambda_j}|^2}{4} +\mathcal{O}(\|A\|^3).
\end{align}

In the general case where one has no knowledge of the eigenstates of $H\in\mathbb{C}^{N\times N}$, a possibility is to let $A$ be a random matrix. More specifically, one could choose
\begin{equation}
    A = \frac{M}{\sqrt{N}},
\end{equation}
where $M$ is sampled from a Gaussian unitary ensemble (GUE). The Gaussian unitary ensemble is a unitarily invariant distribution on Gaussian random matrices such that each matrix element of $M$ satisfies
\begin{equation}
    M_{ij} \sim \left(\mathcal{N}(0, 1/2) +i\mathcal{N}(0, 1/2)\right)(1-\delta_{ij}) + \delta_{ij} \mathcal{N}(0, 1).
\end{equation}
Through a straightforward calculation using unitary invariance of the norm, one can verify that $\mathbb{E}\left[\|A\|\right] = 1$. From this, we can express the expected probability of a cooling step lowering the energy as
\begin{align}
    \mathbb{E}\left[\frac{1}{N}\sum_{i: \lambda_i < \lambda_j} \frac{|\bra{\lambda_i}M\ket{\lambda_j}|^2}{4}\right] &= \frac{1}{4N}\, \sum_{i: \lambda_i < \lambda_j} \mathbb{E}\left[|M_{ij}|^2\right]\\
    &= \frac{|\{i: \lambda_i < \lambda_j\}|}{4N}.
\end{align}
Hence while such a choice of $A$ would work well to initially cool down a random or highly excited state, it becomes less efficient as the state's energy tends closer to the ground state energy.  We see this from the fact that the expected number of trials before a transition is observed scales as
\begin{equation}
    N_{\rm trials} \in \mathcal{O}\left(\frac{N}{|\{i: \lambda_i < \lambda_j\}|} \right).
\end{equation}
Thus at every iteration of the algorithm we have an $\mathcal{O}(1/N)$ probability of transitioning to the ground state, hence
the expected number of iterations to cool into the ground state is $\mathcal{O}(N)$. This suggests that if a GUE random matrix is chosen for $A$ then the cost of preparing the ground state in the worst-case scenario will scale polynomially with the Hilbert space dimension. Despite this, we can still guarantee in such cases that the algorithm will converge to the ground state albeit at an exponential cost. Nonetheless, heuristic choices for $A$ like those inspired by physical phenomena are possible and could be used in many cases to perform the state preparation more efficiently.

\section{Conclusion}

In this work, we introduced a dynamical quantum algorithm for ground-state preparation aimed at overcoming limitations of existing methods. Recent work~\cite{fomichev2023initial} has argued that the most prominent methods for ground-state preparation are based on encoding classical wavefunctions on a quantum computer. In practice, when dealing with large and highly-correlated systems, classical methods can struggle to prepare initial states with significant overlap with the ground state, or even with a low-energy subspace. Hence there is motivation to develop fully-quantum state preparation algorithms, like the one introduced in this paper, that employ dynamic and adaptive approaches to go beyond encoding classical initial states.\\

Moreover, our approach addresses a number of limitations with existing fully quantum methods that use external heat baths to siphon out energy~\cite{shtanko2021algorithms}. Such techniques need additional qubits to represent the bath and often require directly tuning the gaps of the bath to coincide with the eigenvalue gaps of the Hamiltonian. We solve both problems by employing an approach based on quantum signal processing that is able to perturbatively cool the state without any knowledge of the spectral gaps and does not require additional qubits to represent the bath. We find that the number of queries to a quantum simulation oracle for our Hamiltonian that implements $e^{iH}$ scales as $\widetilde{\mathcal{O}}(d^{\,3/2}/\epsilon)$ and the number of queries to a block-encoding of the perturbation $A$ scales as $\widetilde{\mathcal{O}}(d^{\,3/2})$, where $d$ is the maximum number of iterations we wish to carry out and $\epsilon$ is the resolution used to distinguish states of different of energy. This is in contrast to the $\Omega(d^3/\epsilon)$ scaling of bath-based methods, as seen in~\cite{shtanko2021algorithms}.\\

Our work provides a general framework that enables the efficient utilization of any Hermitian operator $A$ as the driving mechanism to systematically steer any input state towards the ground state of any Hamiltonian $H$. Our contributions further include deriving an analytic expression that describes approximately the state's evolution through each dissipation step, which could be of independent interest in estimating leakage and effective simulation under perturbation. While we provide a theoretical selection criteria, discuss specific cases, and suggest potential choices for the interaction operator $A$, the optimal strategy for choosing $A$ to achieve fast convergence remains an open question. This challenge highlights a significant opportunity for future research.

\bibliographystyle{unsrtnat}
\bibliography{main.bib}
\appendix

\section{Proof of Evolution Leakage Bound}\label{appendix:first}
Here we will provide the proof for \cref{thm:hamLeak}
\leaking*
Before we can prove the above leakage theorem, we first need to establish a few lemmas regarding the norms and the leakages of each term in the Dyson series expansion of $e^{-i\Hat{H}t}$ for $\Hat{H} = H_0 + \frac{\sqrt{\delta}}{2}A$. Switching to the interaction picture, we represent $e^{-i\Hat{H}t}$ by $U(t, 0)$ defined by the following recursive expression:
\begin{equation}
    U(t, 0)= \,\mathbb{I} - i \int_{0}^t{dt_1\ A(t_1)U(t_1,0)}
\end{equation}
Where
\begin{equation}
    A(t) = e^{iH_0t}\left(\frac{\sqrt{\delta}}{2}A\right)e^{-iH_0t}
\end{equation}
Repeatedly substituting the expression for $U(t, 0)$ back into itself gives us the Dyson series
\begin{align}\label{eq:bigboy}
U(t, 0) =& \,\mathbb{I} + \frac{\sqrt{\delta}}{2i}\int_{0}^t dt_1e^{iH_0t_1}Ae^{-iH_0t_1}+ \left(\frac{\sqrt{\delta}}{2i}\right)^2\int_{0}^t dt_1 \int_{0}^{t_1} \, dt_2 e^{iH_0t_1}Ae^{-iH_0t_1}e^{iH_0t_2}Ae^{-iH_0t_2} \\
&\cdots+ \left(\frac{\sqrt{\delta}}{2i}\right)^k\int_{0}^t dt_1\int_{0}^{t_1} dt_2 \cdots \int_{0}^{t_{k-1}} dt_k\,e^{iH_0t_1}Ae^{-iH_0t_1} \cdots e^{iH_0t_k}Ae^{-iH_0t_k} +\cdots
\end{align}
We define $U_k(t, 0)$ for convenience as
\begin{equation}
    U_k(t, 0) = \left(\frac{\sqrt{\delta}}{2i}\right)^k\int_{0}^t dt_1\int_{0}^{t_1} dt_2 \cdots \int_{0}^{t_{k-1}} dt_k\,e^{iH_0t_1}Ae^{-iH_0t_1} \cdots e^{iH_0t_k}Ae^{-iH_0t_k}
\end{equation}
Thus
\begin{equation}
    U(t, 0) = \sum_{k=0}^\infty U_k(t, 0)
\end{equation}
We then have the following upper bound on the norm of each term in the series that follows trivially from the triangle inequality and the sub-multiplicative property of the spectral norm.
\begin{lemma}\label{lem:per_term_bound}
Let $A$ be a Hermitian operator with $\|A\|\leq 1$ and $\Pi$ be a projector. Let $\delta \in (0,1)$ and define $\hat{H} = {\rm REF}(\Pi)+ \frac{\sqrt{\delta}}{2}A$. Then $\forall k \in\mathbb{N}^+,\,\forall t\in\mathbb{R}^+$:
\[
\left\|U_k(t, 0)\right\|\leq \frac{t^{k}\cdot \delta^{k/2}}{k!\cdot2^k}
\]
where $U_k(t, 0)$ is the $k^{th}$ order term in the Dyson series expansion of the time evolution operator $e^{-i\hat{H}t} = \sum_{k=0}^\infty U_k(t, 0)$.
\end{lemma}
We then use the above lemma to derive a tighter bound in terms of $t$ on the norm of the transition amplitude induced by each term in the series from any element in the subspace defined by $\Pi$ to any other element in its orthogonal complement $\mathbb{I} - \Pi$.  
\begin{lemma}\label{lem:per_term_violation}
Let $A$ be a Hermitian operator with $\|A\|\leq 1$ and $\Pi$ be a projector. Let $\delta \in (0,1)$ and define $\hat{H} = {\rm REF}(\Pi) + \frac{\sqrt{\delta}}{2}A$. Then $\forall k \in\mathbb{N}^+,\,\forall t\in\mathbb{R}^+$:
\[
\left\|\left(\mathbb{I} - \Pi\right)U_k(t, 0)\,\Pi\right\|\leq \frac{t^{k-1} \delta^{k/2}}{(k-1)!\cdot2}
\]
where $U_k(t, 0)$ is the $k^{th}$ order term in the Dyson series expansion of the time evolution operator $e^{-i\hat{H}t} = \sum_{k=0}^\infty U_k(t, 0)$.
\end{lemma}

\begin{proof}
We will prove this via induction. For the base case of $k  = 1$ we have
\begin{equation}
    \left\|\left(\mathbb{I} - \Pi\right)U_1(t, 0)\,\Pi\right\|
\end{equation}
\begin{equation}
    = \left\|\frac{\sqrt{\delta}}{2i}\int_{0}^t dt_1\left(\mathbb{I} - \Pi\right)e^{iH_0t_1}Ae^{-iH_0t_1}\,\Pi\right\|
\end{equation}
by definition $H_0$ has energy $-1$ with respect to the $\Pi$ subspace and energy $1$ with respect to the $\mathbb{I}-\Pi$ subspace. Therefore, $\left(\mathbb{I} - \Pi\right)e^{iH_0t_1} = e^{it_1}\left(\mathbb{I} - \Pi\right)$ and $e^{-iH_0t_1}\,\Pi = \Pi\,e^{it_1}$. Thus
\begin{align}
     \left\|\left(\mathbb{I} - \Pi\right)U_1(t, 0)\,\Pi\right\| &= \left\|\left(\mathbb{I} - \Pi\right)A\,\Pi\right\|\left|\frac{\sqrt{\delta}}{2i}\int_{0}^t dt_1 e^{i2t_1}\right|\nonumber\\
     &\leq \left|\frac{\sqrt{\delta}}{2i}e^{it}\sin t\right|\nonumber\\
     &\leq \frac{\sqrt{\delta}}{2}
\end{align}
as desired.\\
Moving on to the induction step, we assume that our hypothesis holds for $k\in \mathbb{N}^+$, and we will show that implies that it also holds for $k+1$. Consider
\begin{equation}
    \left\|\left(\mathbb{I} - \Pi\right)U_{k+1}(t, 0)\,\Pi\right\|
\end{equation}
We have $U_{k+1}(t, 0)= \frac{\sqrt{\delta}}{2i}\int_0^t dt_1e^{iH_0t_1} A e^{-iH_0t_1}U_{k}(t_1, 0)$, hence
\begin{equation}
        \left\|\left(\mathbb{I} - \Pi\right)U_{k+1}(t, 0)\,\Pi\right\| =\frac{\sqrt{\delta}}{2}\left\|\int_0^t dt_1 \left(\mathbb{I} - \Pi\right) e^{iH_0t_1} A e^{-iH_0t_1}U_{k}(t_1, 0)\,\Pi\right\|
\end{equation}
Since $\mathbb{I} = \left(\mathbb{I} - \Pi\right) + \Pi$, we write
\begin{align}
    \left\|\left(\mathbb{I} - \Pi\right)U_{k+1}(t, 0)\,\Pi\right\| &= \frac{\sqrt{\delta}}{2}\left\|\int_0^t dt_1 \left(\mathbb{I} - \Pi\right) e^{iH_0t_1} A e^{-iH_0t_1}\left[\left(\mathbb{I} - \Pi\right) + \Pi\right]U_{k}(t_1, 0)\,\Pi\right\|\nonumber\\
    & = \frac{\sqrt{\delta}}{2}\Biggl\|\int_0^t dt_1 \left(\mathbb{I} - \Pi\right) e^{iH_0t_1} A e^{-iH_0t_1}\Pi U_{k}(t_1, 0)\,\Pi\nonumber\\ 
    &\quad + \int_0^t dt_1 \left(\mathbb{I} - \Pi\right) e^{iH_0t_1} A e^{-iH_0t_1}\left(\mathbb{I} - \Pi\right)U_{k}(t_1, 0)\,\Pi\Biggr\|\nonumber\\
    &\leq \frac{\sqrt{\delta}}{2}\left\|\int_0^t dt_1 \left(\mathbb{I} - \Pi\right) e^{iH_0t_1} A e^{-iH_0t_1} \Pi U_{k}(t_1, 0)\,\Pi\right\|\nonumber\\
    &\quad + \frac{\sqrt{\delta}}{2}\left\|\int_0^t dt_1 \left(\mathbb{I} - \Pi\right) e^{iH_0t_1} A e^{-iH_0t_1}\left(\mathbb{I} - \Pi\right)U_{k}(t_1, 0)\,\Pi\right\|\nonumber\\
    &\leq \frac{\sqrt{\delta}}{2}\left\|\int_0^t dt_1 e^{i2t_1}U_{k}(t_1, 0)\,\Pi\right\| + \frac{\sqrt{\delta}}{2}\int_0^t dt_1 \left\|\left(\mathbb{I} - \Pi\right)U_{k}(t_1, 0)\,\Pi\right\|
\end{align}
By our induction hypothesis $\left\|\left(\mathbb{I} - \Pi\right)U_k(t, 0)\,\Pi\right\|\leq \frac{t^{k-1} \delta^{k/2}}{(k-1)!\cdot2}$
\begin{align}
    \frac{\sqrt{\delta}}{2}\int_0^t dt_1 \left\|\left(\mathbb{I} - \Pi\right)U_{k}(t_1, 0)\,\Pi\right\|&\leq \frac{ \delta^{(k+1)/2}}{2\cdot 2}\int_0^t \frac{t_1^{k-1}}{(k-1)!} dt_1\nonumber\\
    &= \frac{1}{2}\cdot \frac{t^k \cdot \delta^{(k+1)/2}}{k!\cdot 2}
\end{align}
Hence it suffices to show
\begin{equation}
    \frac{\sqrt{\delta}}{2}\left\|\int_0^t dt_1 e^{i2t_1}U_{k}(t_1, 0)\,\Pi\right\| \leq \frac{\sqrt{\delta}}{2}\left\|\int_0^t dt_1 e^{i2t_1}U_{k}(t_1, 0)\right\|\leq \frac{1}{2}\cdot \frac{t^k \cdot \delta^{(k+1)/2}}{k!\cdot 2}
\end{equation}
In other words, we need to show
\begin{equation}
    \left\|\int_0^t dt_1 e^{i2t_1}U_{k}(t_1, 0)\right\|\leq  \frac{t^k \cdot \delta^{k/2}}{k!\cdot 2}
\end{equation}
Using integration by parts we get
\begin{align}
    \left\|\int_0^t dt_1 e^{i2t_1}U_{k}(t_1, 0)\right\|  &= \left\|\frac{1}{2i}\left[e^{2it} U_k(t, 0) - \int_0^t e^{2it_1} U'_{k}(t_1, 0) dt_1\right]\right\|\nonumber\\
    &\leq \frac{\left\|U_k(t, 0)\right\| + \left\|\int_0^t e^{2it_1} U'_{k}(t_1, 0) dt_1\right\|}{2}\nonumber\\
    &\leq \frac{\left\|U_k(t, 0)\right\| + \int_0^t \left\|U'_{k}(t_1, 0)\right\| dt_1}{2}
\end{align}
We know $U_{k}(t_1, 0) = F(t_1) - F(0)$ where $F(x) = \frac{\sqrt{\delta}}{2i}\int e^{iH_0x} A e^{-iH_0x}U_{k-1}(x, 0)\,dx$. Then $\left\|U'_{k}(t_1, 0)\right\| = \|F'(t_1)\| \leq \frac{\sqrt{\delta}}{2}\left\|U_{k-1}(t_1, 0)\right\|$. Then
\begin{equation}
    \frac{\left\|U_k(t, 0)\right\| + \int_0^t \left\|U'_{k}(t_1, 0)\right\| dt_1}{2}\leq \frac{\left\|U_k(t, 0)\right\| + \frac{\sqrt{\delta}}{2}\int_0^t \left\|U_{k-1}(t_1, 0)\right\| dt_1}{2}
\end{equation}
Then by \cref{lem:per_term_bound}
\begin{align}
   \left\|\int_0^t dt_1 e^{i2t_1}U_{k}(t_1, 0)\right\| &\leq \frac{\frac{t^{k}\cdot \delta^{k/2}}{k!\cdot2^k} + \int_0^t dt_1\frac{t_1^{k-1}\cdot \delta^{k/2}}{(k-1)!\cdot2^k}}{2}\nonumber\\
    &\leq \frac{t^{k}\cdot \delta^{k/2}}{k!\cdot2^k}\nonumber\\
    &\leq \frac{t^{k}\cdot \delta^{k/2}}{k!\cdot2}
\end{align}
\end{proof}
We then note that when $t$ is an integer multiple of $\pi$, no transitions are allowed between the 2 orthogonal subspaces by $U_1(t, 0)$.  However, transitions are possible at higher orders in the series.
\begin{lemma}\label{lem:first_term_violation}
Let $A$ be a Hermitian operator with $\|A\|\leq 1$ and $\Pi$ be a projector. Let $\delta \in (0,1)$ and define $\hat{H} = {\rm REF}(\Pi) + \frac{\sqrt{\delta}}{2}A$. Then $\forall k \in\mathbb{N}^+$ if $t = k\pi$ :
\[
\left\|\left(\mathbb{I} - \Pi\right)U_1(t, 0)\,\Pi\right\| = 0
\]
Where $U_1(t, 0)$ is the 1st order term in the Dyson series expansion of the time evolution operator $e^{-i\hat{H}t} = \sum_{n=0}^\infty U_n(t, 0)$.
\end{lemma}
\begin{proof}
Let $k\in \mathbb{N}$ and let $t = k\pi$, then we have
\begin{align}
    \left\|\left(\mathbb{I} - \Pi\right)U_1(t, 0)\,\Pi\right\|= \left\|\frac{\sqrt{\delta}}{2i}\int_{0}^t dt_1\left(\mathbb{I} - \Pi\right)e^{iH_0t_1}Ae^{-iH_0t_1}\,\Pi\right\|
\end{align}
by definition $H_0$ has energy $-1$ with respect to the $\Pi$ subspace and energy $1$ with respect to the $\mathbb{I}-\Pi$ subspace. Therefore, $\left(\mathbb{I} - \Pi\right)e^{iH_0t_1} = e^{it_1}\left(\mathbb{I} - \Pi\right)$ and $e^{-iH_0t_1}\,\Pi = \Pi\,e^{it_1}$. Thus
\begin{align}
     \left\|\left(\mathbb{I} - \Pi\right)U_1(t, 0)\,\Pi\right\| &= \left\|\left(\mathbb{I} - \Pi\right)A\,\Pi\right\|\left|\frac{\sqrt{\delta}}{2i}\int_{0}^t dt_1 e^{i2t_1}\right|\nonumber\\
     &\leq \left|\frac{\sqrt{\delta}}{2i}e^{it}\sin t\right|
\end{align}
Then since $\sin(t) = \sin(k\pi) = 0$, we get
\begin{equation}
    \left\|\left(\mathbb{I} - \Pi\right)U_1(t, 0)\,\Pi\right\| = 0
\end{equation}
\end{proof}
We are now finally ready to prove \cref{thm:hamLeak}.
\begin{proof}[Proof of \cref{thm:hamLeak}]
Let $\delta \in (0,1)$ and let $t = \pi\cdot\left\lceil \frac{1}{\pi \sqrt{\delta}}\right\rceil$, then we have
\begin{align}
    \left\|\left(\mathbb{I} - \Pi\right)e^{i\Hat{H}t}\,\Pi\right\| &= \left\|\left(\mathbb{I} - \Pi\right)\left(\sum_{n=0}^\infty U_n(t, 0)\right)\,\Pi\right\|\nonumber\\
    &= \left\|\sum_{n=0}^\infty\left(\mathbb{I} - \Pi\right) U_n(t, 0)\,\Pi\right\|\nonumber\\
    &\leq \sum_{n=0}^\infty\left\|\left(\mathbb{I} - \Pi\right) U_n(t, 0)\,\Pi\right\|
\end{align}
Since $\left(\mathbb{I} - \Pi\right)\Pi = 0$ and by \cref{lem:first_term_violation} $\left(\mathbb{I} - \Pi\right) U_1(t, 0)\,\Pi = 0$, we get
\begin{equation}
    \left\|\left(\mathbb{I} - \Pi\right)e^{i\Hat{H}t}\,\Pi\right\| \leq \sum_{n=2}^\infty\left\|\left(\mathbb{I} - \Pi\right) U_n(t, 0)\,\Pi\right\|
\end{equation}
by \cref{lem:per_term_violation}
\begin{align}
    \left\|\left(\mathbb{I} - \Pi\right)e^{i\Hat{H}t}\,\Pi\right\|&\leq \sum_{n=2}^\infty \frac{t^{n-1} \delta^{n/2}}{(n-1)!\cdot2}\nonumber\\
    &\leq \frac{\sqrt{\delta}}{2}\sum_{n=2}^\infty \frac{1}{(n-1)!}\nonumber\\
    &= \frac{\sqrt{\delta}}{2}\sum_{n=1}^\infty \frac{1}{n!}\nonumber\\
    &= \frac{\sqrt{\delta}}{2} \cdot (e-1)\nonumber\\
    &\leq \sqrt{\delta}
\end{align}
therefore
\begin{equation}
    \left\|\left(\mathbb{I} - \Pi\right)e^{i\Hat{H}t}\,\Pi\right\|^2\leq \delta
\end{equation}
\end{proof}

\section{Proof of Effective Evolution}\label{appendix:second}
Here we will provide the proof for \cref{thm:how_it_be}
\dissipate*

\begin{proof}[Proof of \cref{thm:how_it_be}]
    Writing $e^{-i\Hat{H}t}$ as a Dyson series and recalling $t=1/\sqrt{\delta}$
    \begin{align}
        \Pi e^{-i\Hat{H}t}\,\Pi &= \sum_{k=0}^\infty \Pi U_k(t, 0)\Pi\nonumber\\
        &= \Pi + \frac{\Pi A\Pi}{2i} + \sum_{k=2}^\infty \Pi U_k(t, 0)\Pi\nonumber\\
        &= \Pi + \frac{\Pi A\Pi}{2i} + \sum_{k=2}^\infty\!\left(\frac{\sqrt{\delta}}{2i}\right)^k\int_{0}^t \Pi e^{iH_0t_1}Ae^{-iH_0t_1}dt_1\int_{0}^{t_1}e^{iH_0t_2}Ae^{-iH_0t_2} dt_2 \cdots \int_{0}^{t_{k-1}} e^{iH_0t_k}Ae^{-iH_0t_k}\Pi dt_k
    \end{align}
    For each term in the sum, $U_k(t, 0)$, we can insert $k-1$ resolutions of identity $\mathbb{I} = \left(\mathbb{I} -\Pi\right) + \Pi$ after the first $k-1$ applications of $A$.
    \begin{align}
        \Pi e^{-i\Hat{H}t}\,\Pi &= \Pi + \frac{\Pi A\Pi}{2i} + \sum_{k=2}^\infty\, \Biggl[\,\sum_{\Vec{J}\in \{0,1\}^{k-1}}\, \left(\frac{\sqrt{\delta}}{2i}\right)^k\Pi\, A \left(\prod_{n=1}^{k-1} \left(\mathbb{I} -\Pi\right)^{(1-J_n)} \Pi^{J_n} \,A\right)\Pi\nonumber\\
        &\quad \times \int_{0}^t e^{2i\left(J_1-1\right)t_1}dt_1\int_{0}^{t_1} e^{2i\left(J_2-J_1\right)t_2} dt_2 \cdots \int_{0}^{t_{k-2}} e^{2i\left(J_{k-1}-J_{k-2}\right)t_{k-1}} dt_{k-1} \int_{0}^{t_{k-1}} e^{2i\left(1-J_{k-1}\right)t_k} dt_k \Biggr]
    \end{align}
    Separating the case of $\Vec{J}=\Vec{1}$ from $\Vec{J}\neq\Vec{1}$ (where $\Vec{1}$ is the all ones vector), we get
    \begin{align}\label{eq:fat_boy}
    \begin{split}
        \Pi e^{-i\Hat{H}t}\,\Pi &= \Pi + \frac{\Pi A\Pi}{2i} + \sum_{k=2}^\infty \Biggl[\left(\frac{\Pi\, A\,\Pi\sqrt{\delta}}{2i}\right)^k \,\int_{0}^t dt_1\int_{0}^{t_1} dt_2 \cdots \int_{0}^{t_{k-2}} dt_{k-1} \int_{0}^{t_{k-1}} 1 dt_k \Biggr]\\
        &\quad +\sum_{k=2}^\infty\, \Biggl[\,\sum_{\substack{\Vec{J}\in \{0,1\}^{k-1}\\ \Vec{J}\neq\Vec{1}}}\, \left(\frac{\sqrt{\delta}}{2i}\right)^k\Pi\, A \left(\prod_{n=1}^{k-1} \left(\mathbb{I} -\Pi\right)^{(1-J_n)} \Pi^{J_n} \,A\right)\Pi\\
        &\quad \times \int_{0}^t e^{2i\left(J_1-1\right)t_1}dt_1\int_{0}^{t_1} e^{2i\left(J_2-J_1\right)t_2} dt_2 \cdots \int_{0}^{t_{k-2}} e^{2i\left(J_{k-1}-J_{k-2}\right)t_{k-1}} dt_{k-1} \int_{0}^{t_{k-1}} e^{2i\left(1-J_{k-1}\right)t_k} dt_k \Biggr]\\
        &= \Pi + \frac{\Pi A\Pi}{2i} + \sum_{k=2}^\infty \Biggl[\left(\frac{\Pi\, A\,\Pi}{2i}\right)^k\frac{1}{k!} \Biggr] +\sum_{k=2}^\infty\, \Biggl[\,\sum_{\substack{\Vec{J}\in \{0,1\}^{k-1}\\ \Vec{J}\neq\Vec{1}}}\, \left(\frac{\sqrt{\delta}}{2i}\right)^k\Pi\, A \left(\prod_{n=1}^{k-1} \left(\mathbb{I} -\Pi\right)^{(1-J_n)} \Pi^{J_n} \,A\right)\Pi\\
        &\quad \times \int_{0}^t e^{2i\left(J_1-1\right)t_1}dt_1\int_{0}^{t_1} e^{2i\left(J_2-J_1\right)t_2} dt_2 \cdots \int_{0}^{t_{k-2}} e^{2i\left(J_{k-1}-J_{k-2}\right)t_{k-1}} dt_{k-1} \int_{0}^{t_{k-1}} e^{2i\left(1-J_{k-1}\right)t_k} dt_k \Biggr]\\
        &= e^{-i\Pi(A/2)\Pi}\,\Pi +\sum_{k=2}^\infty\, \Biggl[\,\sum_{\substack{\Vec{J}\in \{0,1\}^{k-1}\\ \Vec{J}\neq\Vec{1}}}\, \left(\frac{\sqrt{\delta}}{2i}\right)^k\Pi\, A \left(\prod_{n=1}^{k-1} \left(\mathbb{I} -\Pi\right)^{(1-J_n)} \Pi^{J_n} \,A\right)\Pi\\
        &\quad \times \int_{0}^t e^{2i\left(J_1-1\right)t_1}dt_1\int_{0}^{t_1} e^{2i\left(J_2-J_1\right)t_2} dt_2 \cdots \int_{0}^{t_{k-2}} e^{2i\left(J_{k-1}-J_{k-2}\right)t_{k-1}} dt_{k-1} \int_{0}^{t_{k-1}} e^{2i\left(1-J_{k-1}\right)t_k} dt_k \Biggr].
        \end{split}
    \end{align}
    We introduce $\Phi_k(\Vec{J})$ in place of the integral expression for brevity
    \begin{align}
        \Phi_k(\Vec{J}) := \int_{0}^t e^{2i\left(J_1-1\right)t_1}dt_1\int_{0}^{t_1} e^{2i\left(J_2-J_1\right)t_2} dt_2 \cdots \int_{0}^{t_{k-2}} e^{2i\left(J_{k-1}-J_{k-2}\right)t_{k-1}} dt_{k-1} \int_{0}^{t_{k-1}} e^{2i\left(1-J_{k-1}\right)t_k} dt_k,
    \end{align}
    which implies
    \begin{align}
        \Biggl\|\Pi e^{-i\Hat{H}t}\,\Pi - e^{-i\Pi(A/2)\Pi}\,\Pi\Biggr\| &= \Biggl\|\sum_{k=2}^\infty\, \sum_{\substack{\Vec{J}\in \{0,1\}^{k-1}\\ \Vec{J}\neq\Vec{1}}}\, \left(\frac{\sqrt{\delta}}{2i}\right)^k\Pi\, A \left(\prod_{n=1}^{k-1} \left(\mathbb{I} -\Pi\right)^{(1-J_n)} \Pi^{J_n} \,A\right)\Pi \cdot\Phi_k(\Vec{J})\Biggr\|\nonumber\\
        &\leq \sum_{k=2}^\infty\,\sum_{\substack{\Vec{J}\in \{0,1\}^{k-1}\\ \Vec{J}\neq\Vec{1}}}\,\Biggl\| \left(\frac{\sqrt{\delta}}{2i}\right)^k\Pi\, A \left(\prod_{n=1}^{k-1} \left(\mathbb{I} -\Pi\right)^{(1-J_n)} \Pi^{J_n} \,A\right)\Pi \cdot\Phi_k(\Vec{J}) \Biggr\|\nonumber\\
        &\leq \sum_{k=2}^\infty\, \sum_{\substack{\Vec{J}\in \{0,1\}^{k-1}\\ \Vec{J}\neq\Vec{1}}}\, \left(\frac{\sqrt{\delta}}{2}\right)^k \cdot\left| \Phi_k(\Vec{J})\right|\nonumber\\
        &\leq \sum_{k=2}^\infty\,\left(\frac{\sqrt{\delta}}{2}\right)^k \cdot 2^{k-1}\max_{\substack{\Vec{J}\in \{0,1\}^{k-1}\\ \Vec{J}\neq\Vec{1}}} \left\{\left| \Phi_k(\Vec{J})\right|\right\}\nonumber\\
        &\leq \frac{1}{2}\sum_{k=2}^\infty\,\delta^{k/2}\cdot\max_{\substack{\Vec{J}\in \{0,1\}^{k-1}\\ \Vec{J}\neq\Vec{1}}} \left\{\left| \Phi_k(\Vec{J})\right|\right\}\label{eq:theugly}.
    \end{align}
    We now bound $\left| \Phi_k(\Vec{J})\right|$ for the case of $\Vec{J}\neq\Vec{1}$. Let $\Vec{J}\in \{0,1\}^{k-1}$ such that $\Vec{J}\neq\Vec{1}$. Then let $c$ be the highest index in $\Vec{J}$ such that $J_c=0$ (that is the first time we cross over to the $(\mathbb{I}-\Pi)$ subspace), we have
    \begin{align}
        \left| \Phi_k(\Vec{J})\right| &= \Biggl|\int_{0}^t e^{2i\left(J_1-1\right)t_1}dt_1\int_{0}^{t_1} e^{2i\left(J_2-J_1\right)t_2} dt_2 \cdots \int_{0}^{t_{c-1}} e^{2i\left(J_{c}-J_{c-1}\right)t_c} dt_c \int_{0}^{t_{c}} e^{2it_{c+1}} \frac{t_{c+1}^{k-c-1}}{(k-c-1)!} \,dt_{c+1}\Biggr|\nonumber\\
        &\leq \int_{0}^t dt_1\int_{0}^{t_1} dt_2 \cdots \int_{0}^{t_{c-1}} dt_c \Biggl|\int_{0}^{t_{c}} e^{2it_{c+1}} \frac{t_{c+1}^{k-c-1}}{(k-c-1)!} \,dt_{c+1}\Biggr|
    \end{align}
    It is easy to check via integration by parts that
    \begin{equation}
        \Biggl|\int_{0}^{t_{c}} e^{2it_{c+1}} \frac{t_{c+1}^{k-c-1}}{(k-c-1)!} \,dt_{c+1}\Biggr|\leq \frac{t_{c}^{k-c-1}}{(k-c-1)!}
    \end{equation}
    Hence the full integral can be bounded by
    \begin{align}
        \left| \Phi_k(\Vec{J})\right| &\leq \int_{0}^t dt_1\int_{0}^{t_1} dt_2 \cdots \int_{0}^{t_{c-1}} dt_c \Biggl|\int_{0}^{t_{c}} e^{2it_{c+1}} \frac{t_{c+1}^{k-c-1}}{(k-c-1)!} \,dt_{c+1}\Biggr|\nonumber\\ 
        &\leq \int_{0}^t dt_1\int_{0}^{t_1} dt_2 \cdots \int_{0}^{t_{c-1}} \frac{t_{c}^{k-c-1}}{(k-c-1)!}\, dt_c\nonumber\\
        &= \frac{t^{k-1}}{(k-1)!}
    \end{align}
    Substituting this upper bound back into \cref{eq:theugly} we get
    \begin{align}
        \Biggl\|\Pi e^{-i\Hat{H}t}\,\Pi - e^{-i\Pi(A/2)\Pi}\,\Pi\Biggr\|&\leq \frac{1}{2}\sum_{k=2}^\infty\,\delta^{k/2}\cdot \frac{t^{k-1}}{(k-1)!}\nonumber\\
    &= \frac{\sqrt{\delta}}{2}\sum_{k=2}^\infty \frac{1}{(k-1)!}\nonumber\\
    &= \frac{\sqrt{\delta}}{2}\sum_{k=1}^\infty \frac{1}{k!}\nonumber\\
    &= \frac{\sqrt{\delta}}{2} \cdot (e-1)\nonumber\\
    &\leq \sqrt{\delta}
    \end{align}
    Hence
    \begin{equation}
        \Biggl\|\Pi e^{-i\Hat{H}t}\,\Pi - e^{-i\Pi(A/2)\Pi}\,\Pi\Biggr\|^2\leq \delta
    \end{equation}
    
\end{proof}

\end{document}